\documentclass{article}
\usepackage{spconf}
\pdfoutput=1
\usepackage{grffile}
\usepackage{epstopdf}
\usepackage{amsmath}
\usepackage{amsfonts}
\usepackage{array}
\usepackage{amsthm}
\usepackage{algorithm}
\usepackage{algpseudocode}
\usepackage{epstopdf}
\usepackage{fixltx2e}
\usepackage{textcomp}
\usepackage{stfloats}
\usepackage{url}
\usepackage{color}
\usepackage{dsfont}
\usepackage{balance}

\usepackage{etoolbox}

\usepackage[demo]{graphicx}
\usepackage{caption}
\usepackage{subcaption}

\newtheorem{lemma}{Lemma}
\AfterEndEnvironment{lemma}{\noindent\ignorespaces}
\newtheorem{theorem}{Theorem}
\AfterEndEnvironment{theorem}{\noindent\ignorespaces}
\newtheorem{definition}{Definition}
\AfterEndEnvironment{definition}{\noindent\ignorespaces}
\AfterEndEnvironment{proof}{\noindent\ignorespaces}

\newcommand{\sqeq}{\medmuskip=2mu \thinmuskip=1mu \thickmuskip=3mu}
\newcommand{\ssqeq}{\medmuskip=1mu \thinmuskip=0mu \thickmuskip=2mu \nulldelimiterspace=-1pt \scriptspace=0pt}

\begin{document}
\ninept
\title{\vspace{-15mm} Recovery of Missing Data in Correlated Smart Grid Datasets}

\name{	
Cristian Genes$^1$,
I\~{n}aki~Esnaola$^{1,3}$,
Samir~M.~Perlaza$^{2,3}$
and Daniel Coca$^1$.
\thanks{Email addresses: c.genes@sheffield.ac.uk, esnaola@sheffield.ac.uk, samir.perlaza@inria.fr, and d.coca@sheffield.ac.uk}
}

\address{$^1$Department of Automatic Control and Systems Engineering, University of Sheffield, UK \\
$^2$Institut National de Recherche en Informatique et en Automatique (INRIA), Lyon, France\\
$^3$Department of Electrical Engineering, Princeton University, NJ, USA.}


%
\setlength\unitlength{1mm}

\newcommand{\insertfig}[3]{
\begin{figure}[htbp]\begin{center}\begin{picture}(120,90)
\put(0,-5){\includegraphics[width=12cm,height=9cm,clip=]{#1.eps}}\end{picture}\end{center}
\caption{#2}\label{#3}\end{figure}}

\newcommand{\insertxfig}[4]{
\begin{figure}[htbp]
\begin{center}
\leavevmode \centerline{\resizebox{#4\textwidth}{!}{\input
#1.pstex_t}}
\caption{#2} \label{#3}
\end{center}
\end{figure}}

\long\def\comment#1{}


\newfont{\bbb}{msbm10 scaled 700}
\newcommand{\CCC}{\mbox{\bbb C}}

\newfont{\bb}{msbm10 scaled 1100}
\newcommand{\CC}{\mbox{\bb C}}
\renewcommand{\SS}{\mbox{\bb S}}
\newcommand{\RR}{\mbox{\bb R}}
\newcommand{\PP}{\mbox{\bb P}}
\newcommand{\QQ}{\mbox{\bb Q}}
\newcommand{\ZZ}{\mbox{\bb Z}}
\newcommand{\FF}{\mbox{\bb F}}
\newcommand{\GG}{\mbox{\bb G}}
\newcommand{\EE}{\mbox{\bb E}}
\newcommand{\NN}{\mbox{\bb N}}
\newcommand{\KK}{\mbox{\bb K}}


\newcommand{\av}{{\bf a}}
\newcommand{\bv}{{\bf b}}
\newcommand{\cv}{{\bf c}}
\newcommand{\dv}{{\bf d}}
\newcommand{\ev}{{\bf e}}
\newcommand{\fv}{{\bf f}}
\newcommand{\gv}{{\bf g}}
\newcommand{\hv}{{\bf h}}
\newcommand{\iv}{{\bf i}}
\newcommand{\jv}{{\bf j}}
\newcommand{\kv}{{\bf k}}
\newcommand{\lv}{{\bf l}}
\newcommand{\mv}{{\bf m}}
\newcommand{\nv}{{\bf n}}
\newcommand{\ov}{{\bf o}}
\newcommand{\pv}{{\bf p}}
\newcommand{\qv}{{\bf q}}
\newcommand{\rv}{{\bf r}}
\newcommand{\sv}{{\bf s}}
\newcommand{\tv}{{\bf t}}
\newcommand{\uv}{{\bf u}}
\newcommand{\wv}{{\bf w}}
\newcommand{\vv}{{\bf v}}
\newcommand{\xv}{{\bf x}}
\newcommand{\yv}{{\bf y}}
\newcommand{\zv}{{\bf z}}
\newcommand{\ellv}{{\bf \ell}}
\newcommand{\zerov}{{\bf 0}}
\newcommand{\onev}{{\bf 1}}


\newcommand{\Am}{{\bf A}}
\newcommand{\Bm}{{\bf B}}
\newcommand{\Cm}{{\bf C}}
\newcommand{\Dm}{{\bf D}}
\newcommand{\Em}{{\bf E}}
\newcommand{\Fm}{{\bf F}}
\newcommand{\Gm}{{\bf G}}
\newcommand{\Hm}{{\bf H}}
\newcommand{\Id}{{\bf I}}
\newcommand{\Jm}{{\bf J}}
\newcommand{\Km}{{\bf K}}
\newcommand{\Lm}{{\bf L}}
\newcommand{\Mm}{{\bf M}}
\newcommand{\Nm}{{\bf N}}
\newcommand{\Om}{{\bf O}}
\newcommand{\Pm}{{\bf P}}
\newcommand{\Qm}{{\bf Q}}
\newcommand{\Rm}{{\bf R}}
\newcommand{\Sm}{{\bf S}}
\newcommand{\Tm}{{\bf T}}
\newcommand{\Um}{{\bf U}}
\newcommand{\Wm}{{\bf W}}
\newcommand{\Vm}{{\bf V}}
\newcommand{\Xm}{{\bf X}}
\newcommand{\Ym}{{\bf Y}}
\newcommand{\Zm}{{\bf Z}}


\newcommand{\Ac}{{\cal A}}
\newcommand{\Bc}{{\cal B}}
\newcommand{\Cc}{{\cal C}}
\newcommand{\Dc}{{\cal D}}
\newcommand{\Ec}{{\cal E}}
\newcommand{\Fc}{{\cal F}}
\newcommand{\Gc}{{\cal G}}
\newcommand{\Hc}{{\cal H}}
\newcommand{\Ic}{{\cal I}}
\newcommand{\Jc}{{\cal J}}
\newcommand{\Kc}{{\cal K}}
\newcommand{\Lc}{{\cal L}}
\newcommand{\Mc}{{\cal M}}
\newcommand{\Nc}{{\cal N}}
\newcommand{\Oc}{{\cal O}}
\newcommand{\Pc}{{\cal P}}
\newcommand{\Qc}{{\cal Q}}
\newcommand{\Rc}{{\cal R}}
\newcommand{\Sc}{{\cal S}}
\newcommand{\Tc}{{\cal T}}
\newcommand{\Uc}{{\cal U}}
\newcommand{\Wc}{{\cal W}}
\newcommand{\Vc}{{\cal V}}
\newcommand{\Xc}{{\cal X}}
\newcommand{\Yc}{{\cal Y}}
\newcommand{\Zc}{{\cal Z}}


\newcommand{\alphav}{\hbox{\boldmath$\alpha$}}
\newcommand{\betav}{\hbox{\boldmath$\beta$}}
\newcommand{\gammav}{\hbox{\boldmath$\gamma$}}
\newcommand{\deltav}{\hbox{\boldmath$\delta$}}
\newcommand{\etav}{\hbox{\boldmath$\eta$}}
\newcommand{\lambdav}{\hbox{\boldmath$\lambda$}}
\newcommand{\epsilonv}{\hbox{\boldmath$\epsilon$}}
\newcommand{\nuv}{\hbox{\boldmath$\nu$}}
\newcommand{\muv}{\hbox{\boldmath$\mu$}}
\newcommand{\zetav}{\hbox{\boldmath$\zeta$}}
\newcommand{\phiv}{\hbox{\boldmath$\phi$}}
\newcommand{\psiv}{\hbox{\boldmath$\psi$}}
\newcommand{\thetav}{\hbox{\boldmath$\theta$}}
\newcommand{\tauv}{\hbox{\boldmath$\tau$}}
\newcommand{\omegav}{\hbox{\boldmath$\omega$}}
\newcommand{\xiv}{\hbox{\boldmath$\xi$}}
\newcommand{\sigmav}{\hbox{\boldmath$\sigma$}}
\newcommand{\piv}{\hbox{\boldmath$\pi$}}
\newcommand{\rhov}{\hbox{\boldmath$\rho$}}

\newcommand{\Gammam}{\hbox{\boldmath$\Gamma$}}
\newcommand{\Lambdam}{\hbox{\boldmath$\Lambda$}}
\newcommand{\Deltam}{\hbox{\boldmath$\Delta$}}
\newcommand{\Sigmam}{\hbox{\boldmath$\Sigma$}}
\newcommand{\Phim}{\hbox{\boldmath$\Phi$}}
\newcommand{\Pim}{\hbox{\boldmath$\Pi$}}
\newcommand{\Psim}{\hbox{\boldmath$\Psi$}}
\newcommand{\Thetam}{\hbox{\boldmath$\Theta$}}
\newcommand{\Omegam}{\hbox{\boldmath$\Omega$}}
\newcommand{\Xim}{\hbox{\boldmath$\Xi$}}

\newcommand{\supp}{{\hbox{supp}}}
\newcommand{\sinc}{{\hbox{sinc}}}
\newcommand{\diag}{{\hbox{diag}}}
\renewcommand{\det}{{\hbox{det}}}
\newcommand{\trace}{{\hbox{tr}}}
\newcommand{\sign}{{\hbox{sign}}}
\renewcommand{\arg}{{\hbox{arg}}}
\newcommand{\var}{{\hbox{var}}}
\newcommand{\cov}{{\hbox{cov}}}
\newcommand{\SINR}{{\sf SINR}}
\newcommand{\SNR}{{\sf SNR}}
\newcommand{\Ei}{{\rm E}_{\rm i}}
\renewcommand{\Re}{{\rm Re}}
\renewcommand{\Im}{{\rm Im}}
\newcommand{\eqdef}{\stackrel{\Delta}{=}}
\newcommand{\defines}{{\,\,\stackrel{\scriptscriptstyle \bigtriangleup}{=}\,\,}}
\newcommand{\<}{\left\langle}
\renewcommand{\>}{\right\rangle}
\newcommand{\herm}{{\sf H}}
\newcommand{\transp}{{\sf T}}
\renewcommand{\vec}{{\rm vec}}


\newcommand{\GameNF}{\mathcal{G} = \left(\mathcal{K}, \left\lbrace\mathcal{A}_k \right\rbrace_{k \in \mathcal{K}},\phi \right)}
\newcommand{\gameNF}{\mathcal{G}}
\newcommand{\BR}{\mathrm{BR}}

\maketitle
\begin{abstract}
We study the recovery of missing data from multiple smart grid datasets within a matrix completion framework. The datasets contain the electrical magnitudes required for monitoring and control of the electricity distribution system. Each dataset is described by a low rank matrix. Different datasets are correlated as a result of containing measurements of different physical magnitudes generated by the same distribution system. To assess the validity of matrix completion techniques in the recovery of missing data, we characterize the fundamental limits when two correlated datasets are jointly recovered. We then proceed to evaluate the performance of Singular Value Thresholding (SVT) and Bayesian SVT (BSVT) in this setting. We show that BSVT outperforms SVT by simulating the recovery for different correlated datasets. The performance of BSVT displays the tradeoff behaviour described by the fundamental limit, which suggests that BSVT exploits the correlation between the datasets in an efficient manner.

\end{abstract}

\section{Introduction}
The integration of residential low carbon energy sources such as solar or wind power generates bidirectional power flows that affect the stability of the smart grid \cite{WSDBK08}. The control strategies need to adapt to the new challenges posed by the additional distributed energy sources. In this context, the monitoring procedures are expected to manage the dynamic and unknown scenarios and to provide timely and accurate data describing the state of the grid. 
For example, the lack of data quality in power systems contributed towards several large-scale blackouts such as the 2003 U.S.-Canadian blackout \cite{us2004final} and the 2003 Italy blackout \cite{ucte2004final}.
In addition, the integration of the Internet of things into the smart grid will significantly increase the number of datasets \cite{BCCZ12}. In practical scenarios, state estimation and monitoring systems face challenges like data injection attacks \cite{LNR11}, \cite{KP11}, \cite{KJTT11}, \cite{OEVKP16}, \cite{SEPP17} or missing data \cite{GWGCFS16}, \cite{GEPOC16}, \cite{GEPOC17}. Telemetry errors such as sensor failures or communication issues lead to incomplete sets of observations that do not fully describe the state of the grid. Therefore, it is vital to estimate the missing data based on the available observations. For instance, accurate measurements are necessary to implement centralized control schemes for voltage regulation in distribution systems \cite{IYFIOOH16}.

Matrix completion (MC) is proposed in \cite{CR09} as technique to recover missing data from partial observations. MC-based recovery exploits the fact that correlated state variable vectors give rise to approximately low rank data matrices. Specifically, in a convex optimization context, a low rank matrix is estimated given that a sufficient fraction of the entries is observed. See for instance \cite{CP10} and \cite{CT10}. However, when the number of observations is insufficient, the recovery of the data matrix is not possible. A potential way forward in this case is to attempt a joint recovery of multiple datasets by exploiting the fact that when datasets are correlated the rank of the resulting joint dataset grows in a sub-additive fashion. When the number of observations in one dataset is limited, this approach allows the estimation process to incorporate datasets produced by other sources in the system.

A framework for jointly recovering multiple datasets is provided in \cite{LMWY_13} where the MC setting is extended to the tensor case. Moreover, the singular value decomposition is extended to the tensor case in \cite{KBHH13} which leads to the development of a tensor nuclear norm based algorithm in \cite{ZEAHK14}. Alternatively, a collective MC framework is proposed in \cite{GYYC15} to exploit the correlation between matrices with shared structure. However, the common structure constraint does not allow for sufficient generality in the definition of the correlation structure between datasets in a smart grid context.

This paper proposes an estimation setting in which data from multiple datasets is combined into a single data matrix that is recovered using MC-based algorithms. This allows the recovery process to exploit not only correlations within a dataset but also between datasets in the joint estimation paradigm. Specifically, the correlation between datasets is leveraged to facilitate the recovery when the number of observations in one dataset is limited.
In addition, the fundamental limit of the joint recovery setting for two correlated datasets is characterized within an MC framework and, based on the geometry dictated by the fundamental limit, the joint recovery performance of two MC-based algorithms is benchmarked for different levels of correlation between the combined datasets. Numerical results show that the recently proposed Bayesian Singular Value Theresholding (BSVT) algorithm \cite{GEPOC17} is more effective in exploiting the correlation between datasets when compared to the Singular Value Theresholding (SVT) algorithm \cite{CCS10}.
\section{System Model}

Consider an electricity distribution system with $N$ low voltage (LV) feeders.  At the head of each feeder, a sensing unit measures various electrical magnitudes, e.g., voltage, intensity, active and reactive power at given time instants. These measures comprise the state variables that the operator uses for control, monitoring, and management purposes.  The set of observations available to the operator is incomplete and corrupted by noise. The operator estimates the missing data based on the available observations.
In the following, the analysis is carried out for a particular electrical magnitude, i.e., phase voltage.
\subsection{Source Model}

For a given phase voltage state variable, let $m_{i,j}^{(s)}$ be the corresponding value on phase $s \in \{A, B, C\}$, at feeder $i \in \{1,2, ...,N\}$ and time $j \in \{1,2, ...,M\}$. 
The matrix with the measurements for phase $s$, denoted by $\Mm^{(s)}\in \mathbb{R}^{M \times N}$, contains the aggregated measurement vectors from all feeders
\begin{equation} \label{eq:state_matrix}
\Mm^{(s)}\eqdef[\mv_{1}^{(s)}, \mv_{2}^{(s)}, ...,\mv_{N}^{(s)}],
\end{equation}
where the measurement vectors are given by
\begin{equation} \label{eq:state_vector}
\mv_{i}^{(s)}\eqdef [m_{i,1}^{(s)}, m_{i,2}^{(s)}, ...,m_{i,M}^{(s)}]^{\sf T} \in \mathbb{R}^{M}.
\end{equation}
The resulting data matrices $\Mm^{(\text{A})}$, $\Mm^{(\text{B})}$, $\Mm^{(\text{C})}$ contain the voltage measurements on phase A, B and C respectively, at time instants $1,2,...,M$ for all $N$ feeders. 
\subsection{Real data model}

Real data collected as part of the {``Low Voltage Network Solutions"} project run by Electricity North West Limited (ENWL) \cite{enwl}, is used in the following to model the statistical structure of the random process governing the phase voltage state variables. The dataset contains voltage measurements of phases A, B and C collected from 200 residential secondary substations across North West of England from June 2013 to January 2014. Each substation generates a daily file that contains the voltage measurements on all three phases. 

An analysis of the distribution and sample covariance matrix of the phase A voltage measurements in the LV dataset under consideration is presented in \cite{GEPOC16}. Therein, it is shown that voltage measurements can be modelled as a multivariate Gaussian random process for $i \in \{1,2, ...,N\}$. Specifically, we model the voltage measurements as 
\begin{equation} \label{eq:Gmodel}
\mv_i^{(s)} {\sim} \Nc (\muv_s,{\Sigmam_s}),
\end{equation}
and ${\mv_i}$ for $i \in \{1,2, ...,N\}$, is a sequence of independent and identically distributed random variables.
Moreover, it is also shown in \cite{GEPOC16} that the sample covariance matrix for phase A exhibits a structure that is approximately Toeplitz. In addition, because the voltage data is correlated, the covariance matrix displays a high correlation across feeders and time instants. It is shown in \cite{GEPOC17} that the singular value decomposition of a $500 \times 500$ matrix with phase A voltage measurements has a large condition number \cite{seber}.


\begin{figure}[!t]
\centering
\subcaptionbox{Sample covariance matrix of the phase B voltage data matrix. \label{fig:vb_cov}}
{\includegraphics[width=0.23\textwidth]{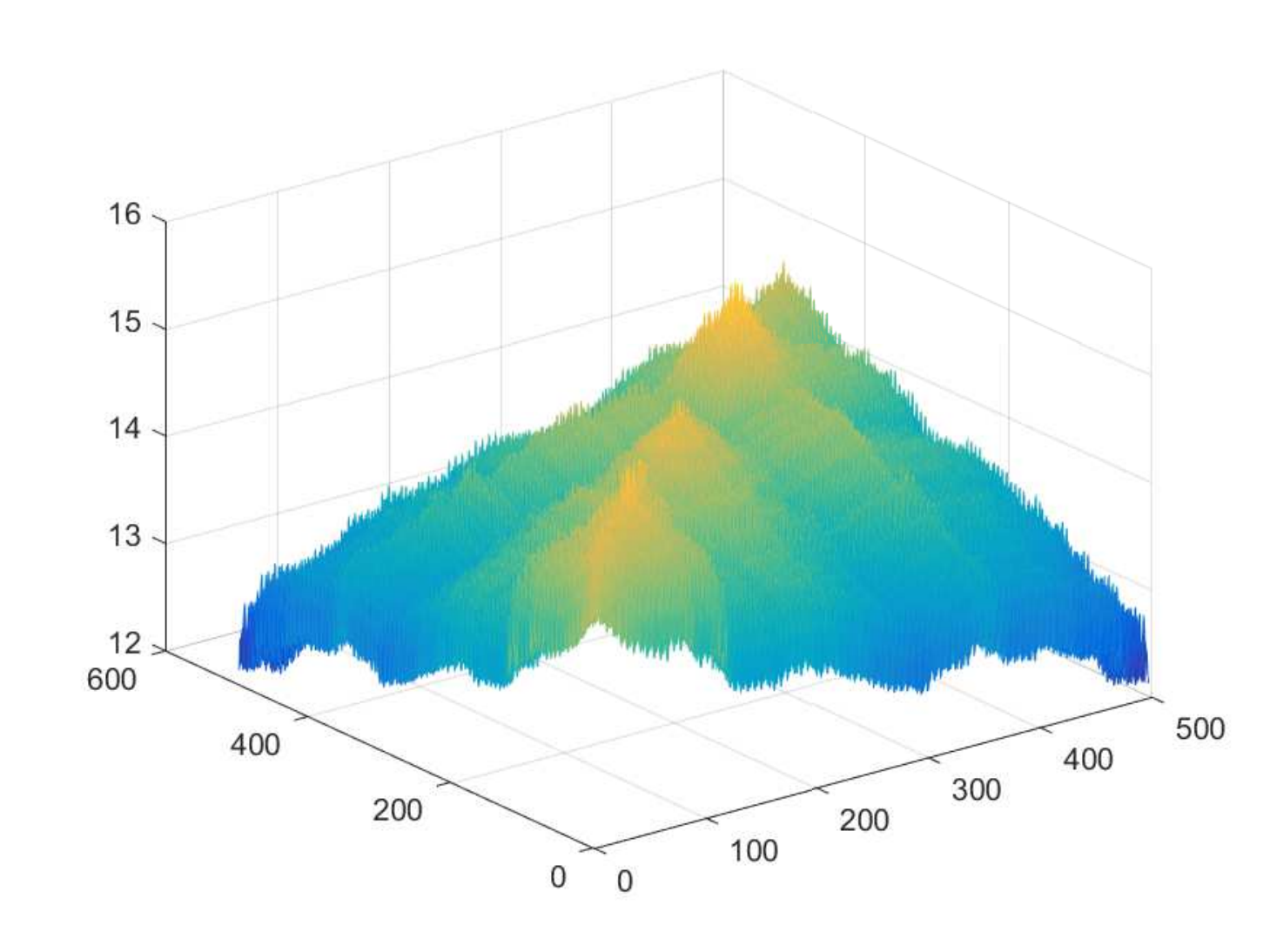}}
\hfill
\subcaptionbox{Sample covariance matrix of the combined phase B and C voltage data matrices. \label{fig:vbc_cov}}
{\includegraphics[width=0.23\textwidth]{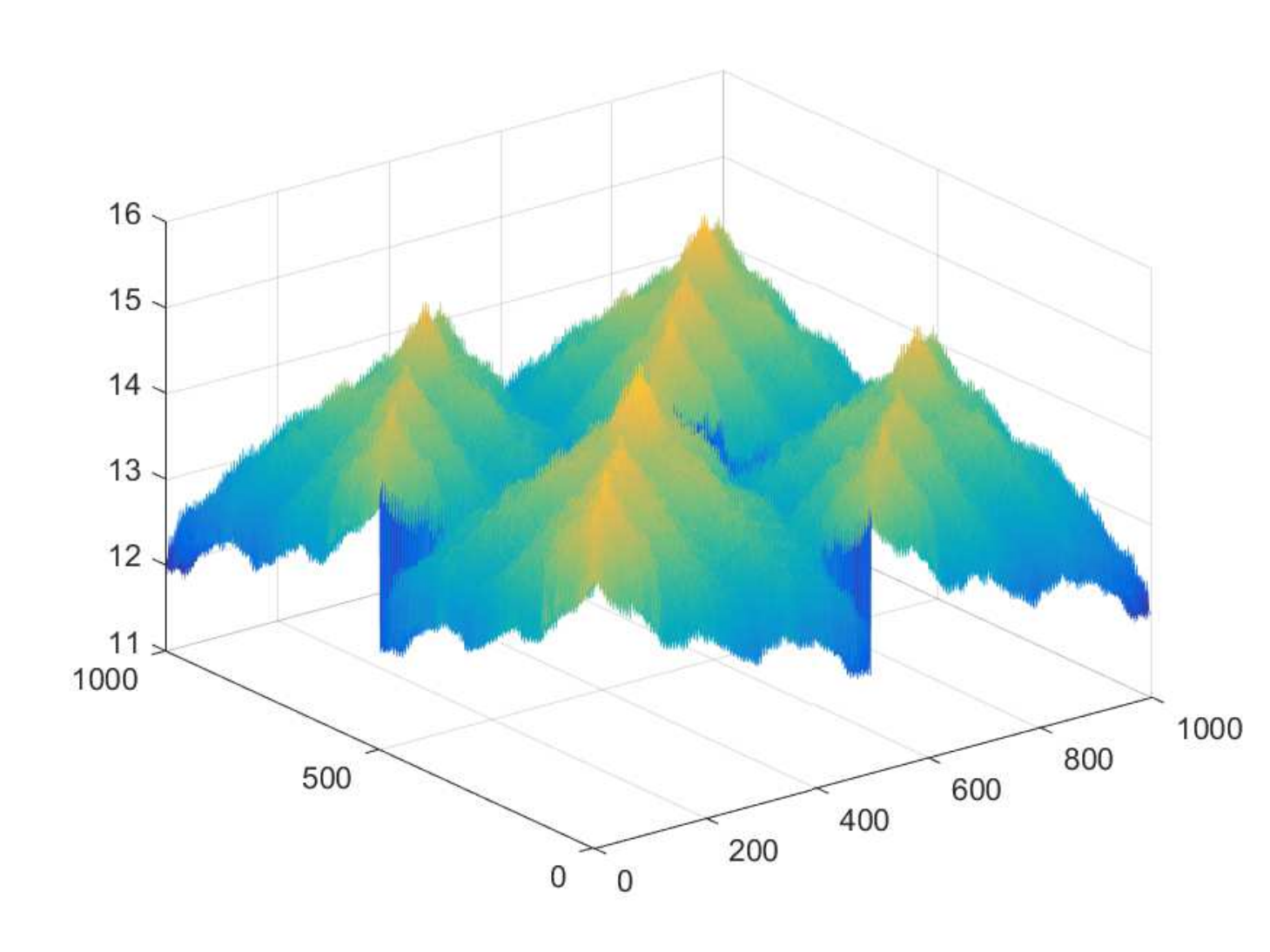}}
\caption{Sample covariance matrices obtained using the real data provided by ENWL.}
\vspace{-5mm}
\end{figure}

Part of the LV data collected by ENWL is used to construct two complete data matrices $\Mm^{(\text{B})}$ and $\Mm^{(\text{C})}$ with $M=N=500$ that contain phase B and phase C voltage measurements from the LV grid. The sample covariance matrix of the data matrix $\Mm^{(\text{B})}$ is depicted in Fig. \ref{fig:vb_cov}. 
As expected, and in agreement with the observation in \cite{GEPOC16}, the sample covariance matrix for the phase voltage data exhibits a structure that is approximately Toeplitz.
In addition, when the phase B and phase C data matrices are combined into a single data matrix, i.e., $\Mm^{(\text{BC})}=[\Mm^{(\text{B})},\, \Mm^{(\text{C})}]^{\sf T}$, the resulting sample covariance matrix is depicted in Fig. \ref{fig:vbc_cov}. Interestingly, the sample covariance for the combined matrix is a block matrix with four elements where each element exhibits a structure that is approximately Toeplitz. Based on this observation, the following section proposes a general model for correlated voltage datasets generated by different phases in smart grid systems.
\subsection{Synthetic data model} \label{sec:syn_data_model}

A mathematical description of the model used to generate two correlated synthetic datasets follows. Let us denote the data matrix for the first dataset by $\Mm_1 \in \mathbb{R}^{M \times N}$ and the data matrix for the second dataset by $\Mm_2\in \mathbb{R}^{M \times N}$. In this setting, the combined matrix is denoted by $\Mm \in \mathbb{R}^{2M \times N}$ given by
\begin{equation}
\Mm\eqdef
\left[
\begin{array}{c}
\Mm_1\\
\Mm_2
\end{array}
\right].
\end{equation}     
Hence, the combined state variable matrix is defined as
\begin{equation} 
\begin{split}
\Mm=[\mv_1, \mv_2, ..., \mv_N],
\end{split}
\end{equation}
where each state variable vector $\mv_i \in \mathbb{R}^{2M}$ for $i \in \{1,2, ..., N\}$ is generated by a multivariate Gaussian process with $\mathbf{0}$ mean and covariance matrix $\Sigmam$, i.e.,
\begin{equation} 
\begin{split}
\mv_i {\sim} \Nc(\mathbf{0},\mathbf{\Sigma}).
\end{split}
\end{equation}
The covariance matrix $\Sigmam$ is a block matrix in which block $\Sigmam_{ll}$ is a Toeplitz matrix describing the covariance matrix of the dataset $l \in \{1,2\}$. The resulting covariance matrix is given by
\begin{equation} \label{eq:synth_data_model}
\Sigmam \eqdef
\left[
\begin{array}{c c}
\Sigmam_{11} & \psi \Sigmam_{11} \\
\psi \Sigmam_{11}  & \Sigmam_{22} \\
\end{array}
\right],
\end{equation}
where $\Sigmam_{ll} \in \mathbb{R}^{M \times M}$ and $\psi \in [0,1]$.
In this framework, the elements of $\Sigmam_{ll}$ are defined as
\begin{equation} 
\begin{split}
(\Sigmam_{ll})_{i,j} \eqdef \rho^{\frac{1}{\zeta_{ll}}|i-j|},
\end{split}
\end{equation}
where $(\Sigmam_{ll})_{i,j}$ denotes the entry in row $i$ and column $j$ of the matrix $\Sigmam_{ll}$ with $i \in \{1, 2,  \ldots, M\}$, $j \in \{1, 2, \ldots, M\}$, $\rho \in (0,1)$ and $\zeta_{ll}$ a design parameter. Hence, the matrix $\Sigmam_{ll}$ is given by
\begin{equation}  \label{eq:sigma_ll}
\Sigmam_{ll}=\textnormal{Toeplitz}(1, \ldots, \upsilon_{ll}),
\end{equation}
where $\upsilon_{ll} \in [0,1)$ obeys
\begin{equation}  \label{eq:end_corr}
\begin{split}
\upsilon_{ll} = \rho^{\frac{1}{\zeta_{ll}}(M-1)}.
\end{split}
\end{equation}

\begin{figure}[t!]
\centering
\includegraphics[width=.48\textwidth]{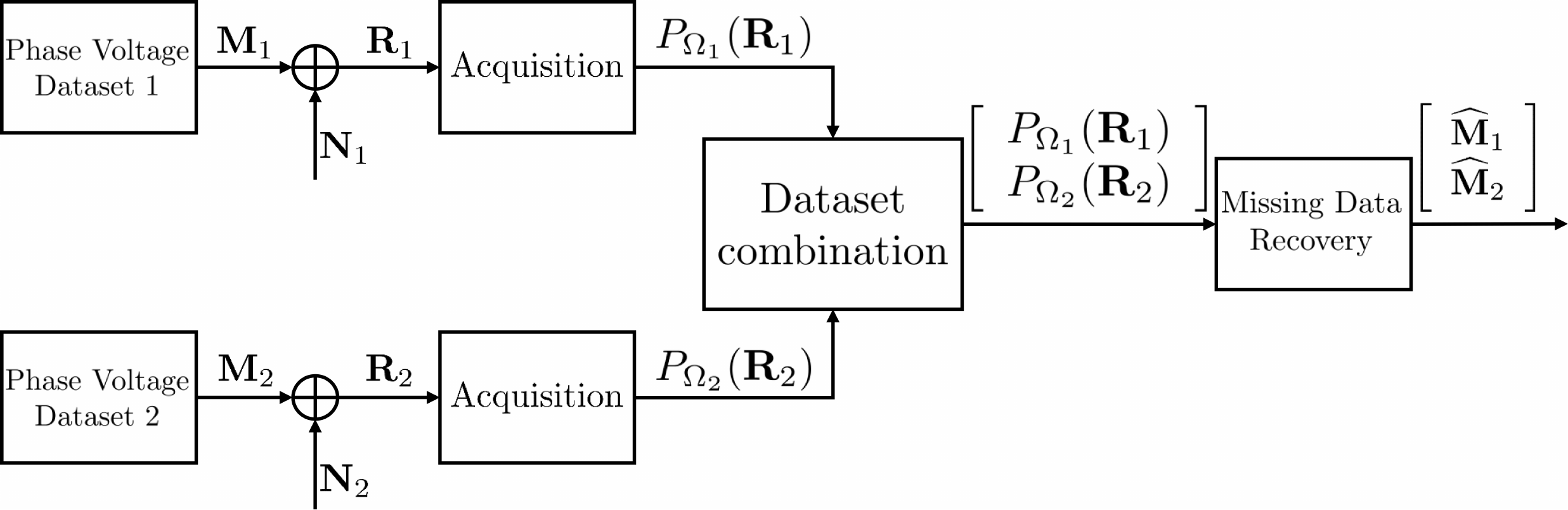}
\caption{Block diagram describing the system model for the joint recovery of two datasets.}
\label{fig:block}
\vspace{-5mm}
\end{figure}

Fig. \ref{fig:block} describes the system model for the joint recovery of two datasets produced in an LV distribution system.
In this setting, each phase voltage data matrix fully describes the state of the grid over $M$ time instants and across $N$ feeders. However, in the acquisition process, part of the measurements are lost and the ones that are available are corrupted by noise. The missing data recovery strategy needs to estimate the actual state of the grid for a noisy subset of observations.
\subsection{Acquisition}

The phase voltage measurements are assumed to be corrupted by additive white Gaussian noise (AWGN) such that for each dataset the resulting observations are given by
\begin{equation}
\Rm_l=\Mm_l+\Nm_l,
\end{equation}
where $l \in \{1,2\}$ denotes the number of datasets and
\begin{equation} \label{eq:noise}
(\Nm_l)_{i,j} \sim \Nc(0,\sigma^2_{\Nm_l}),
\end{equation}
where $i \in \{1,2,...,M\}$ and $j \in \{1,2,...,N\}$. Moreover, it is also assumed that only a fraction of the complete set of observations (entries in $\Rm_l$) are communicated to the operator. Denote by $\Omega_l$ the subset of observed entries of the dataset $l$, i.e.,
\begin{equation}
\Omega_l \eqdef \{(i,j) : (\Rm_l)_{i,j} \textnormal{ is observed}\}.
\end{equation}
Formally, the acquisition process is modelled by the functions $P_{\Omega_l} :  \mathbb{R}^{M \times N} \to  \mathbb{R}^{M \times N}$ with $l \in \{1,2\}$ and
\begin{equation} \label{eq:obs}
P_{\Omega_l}(\Rm_l)=
\begin{cases}
(\Rm_l)_{i,j}, \quad &(i,j) \in \Omega_l, \\
0, &\textnormal{otherwise}.
\end{cases}
\end{equation}
The observations given by (\ref{eq:obs}) describe all the data from dataset $l$ that is available to the operator for estimation purposes. Therefore, the recovery of the missing data is performed from the observations $P_{\Omega_l}(\Rm_l)$.
As depicted in Fig. \ref{fig:block}, the acquisition step is performed independently for each dataset. After the acquisition step, the available observations from each dataset are combined into a single data matrix, i.e., 
\begin{equation}
P_{\Omega}(\Rm)=
\left[
\begin{array}{c}
P_{\Omega_1}(\Rm_1)\\
P_{\Omega_2}(\Rm_2)
\end{array}
\right] \in \mathbb{R}^{2M \times N},
\end{equation}
where $\Omega$ denotes the combined set of available observations from the two datasets. The resulting matrix $P_{\Omega}(\Rm)$ is used for estimation purposes in the joint recovery paradigm.
\subsection{Estimation}

The estimation process for the combined matrix of measurements, based on the available observations from each dataset is modelled by the function $g:\mathbb{R}^{2M \times N} \to \mathbb{R}^{2M \times N}$, where $\Omega$ denotes the combined set of available observations from both datasets.
The estimate $\widehat{\Mm}=g\big(P_{\Omega_1}(\Rm_1),P_{\Omega_2}(\Rm_2)\big)$ is obtained by solving an optimization problem based on a given optimality criterion. In the following, the optimality criterion is the normalized mean square error (NMSE) given by
\begin{equation} \label{eq:nmse}
\textnormal{NMSE} \left( \Mm;g \right)= \frac{\mathbb{E}\left[ \|\Mm-g\big(P_{\Omega_1}(\Rm_1),P_{\Omega_2}(\Rm_2)\big)\|^2_F\right ]}{\|\Mm\|_F^2},
\end{equation}
where $\|\cdot\|_F$ denotes the Frobenius norm. 
\section{Recovering missing data using matrix completion}

Given a matrix $\Mm$ of size $2M \times N$, and observations $P_{\Omega}(\Mm)$, the recovery of the missing entries is not feasible in the general case. However, when $\Mm$ is low rank or approximately low rank, it is shown in \cite{CR09} that if the entries on $\Omega$ are sampled uniformly at random, the missing entries are recovered with high probability by solving the following optimization problem:
\begin{equation} \label{convex_lr}
\begin{aligned}
& {\underset{\mathbf{X}}{\text{minimize}}}
& & \mathrm \lVert \mathbf{X} \rVert_{*} \\
& \text{subject to}
& & P_{\Omega}(\mathbf{X})=P_{\Omega}(\mathbf{M}),\\
\end{aligned}
\end{equation}
where $\lVert \mathbf{X} \rVert_{*}$ denotes the nuclear norm of the matrix $\Xm$. To simplify the notation, let us assume that $2M \geq N$. We proceed to present the two MC-based algorithms used to assess the joint recovery performance. Namely, the SVT algorithm proposed in \cite{CCS10} and the BSVT approach presented in \cite{GEPOC17}.
\subsection{Singular Value Theresholding}

SVT is an MC-based algorithm \cite{CCS10} which produces a sequence of matrices $\Xm^{(k)}$ that converges to the unique solution of the following optimization problem:
\begin{equation} \label{svt}
\begin{aligned}
& {\underset{\mathbf{X}}{\text{minimize}}}
& & \tau \mathrm  \lVert \mathbf{X} \rVert_{*}+ \frac{1}{2}\lVert \mathbf{X} \rVert_{F}^{2} \\
& \text{subject to}
& & P_{\Omega}(\mathbf{X})=P_{\Omega}(\mathbf{M}),\\
\end{aligned}
\end{equation}
Note that when $\tau \to \infty$, the optimization problem in (\ref{svt}) converges to the nuclear norm minimization problem in (\ref{convex_lr}). The iterations of the SVT algorithm are:
\begin{equation} \label{eq:inter_svt}
\begin{cases}
\mathbf{X}^{(k)}= D_{\tau}(\mathbf{Y}^{(k-1)}), \\
\mathbf{Y}^{(k)}=\mathbf{Y}^{(k-1)}+\delta_s \big( P_{\Omega}(\mathbf{M})-P_{\Omega}(\mathbf{X}^{(k)}) \big),\\
\end {cases}
\end{equation}
where $\mathbf{Y}^{(0)}=\mathbf{0}$ is used for initialization, $\delta_s$ is the step size that obeys $0 < \delta_s < 2$, and the soft-thresholding operator, $D_{\tau}$ that shrinks the singular values of $\mathbf{Y}^{(k-1)}$ towards zero \cite{CCS10}. 

Interestingly, the choice of $\tau$ is important to guarantee a successful recovery, since large values guarantee a low-rank matrix estimate but for values larger than $\underset{i}{\textnormal{max}} \,( \sigma_i(\Ym))$ all the singular values vanish. In \cite{CCS10}, the proposed threshold is $\tau =5N$.
However, simulation results presented in \cite{GEPOC16} show that $\tau =5N$ gives suboptimal performance when the number of missing entries is large.
The main shortcoming of the SVT algorithm is the lack of guidelines for tuning the threshold $\tau$. This problem is addressed in \cite{GEPOC17} where a new algorithm is proposed to adapt the recovery to the dataset by leveraging knowledge of the second order statistics.
\subsection{Bayesian Singular Value Theresholding}

BSVT is an MC-based algorithm \cite{GEPOC17} that is able to optimize the value of $\tau$ at each iteration using additional prior knowledge in the form of second order statistics. The optimization of the soft-theresholding step is performed using Stein\textquotesingle s unbiased risk estimate (SURE) \cite{stein81} for which a closed-form expression is presented in \cite{CST13}. However, the result therein pertains to input matrices $\mathbf{Z}$ that accept the following model:
\begin{equation} \label{eq:sure_model}
\Zm=\Mm+\Wm,
\end{equation}
where the entries of $\Wm$ are
\begin{equation}
(\Wm)_{i,j} \overset{iid}{\sim} \Nc(0,\sigma_{\Zm}^2),
\end{equation}
where $\sigma_{\mathbf{Z}}^2$ is the variance of the $(\mathbf{W})_{i,j}$ entries with $\ssqeq i\in\{1,2, \ldots, 2M\}$ and $\ssqeq j\in\{1,2, \ldots, N\}$.
Using the prior knowledge in the form of the second order statistics, the BSVT algorithm, computes the matrix $\Zm$ at iteration $k$, i.e., $\Zm^{(k)}$, as
\begin{equation}
\Zm^{(k)}=\Ym^{(k)}+\Lm^{(k)},
\end{equation}
where $\Ym^{(k)}$ is defined in (\ref{eq:inter_svt}) and $\Lm^{(k)}$ is the linear minimum mean square error (LMMSE) estimate.
Consequently, the incorporation of the LMMSE step into the structure of the BSVT algorithm facilitates the use of SURE \cite{stein81} which is given by
\begin{equation} \label{sure}
\sqeq
\begin{split}
\textnormal{SURE}(D_{\tau})(\Zm) = &-2MN \sigma_{\Zm}^2 + \sum_{i=1}^{N} \textnormal{min}(\tau^2,\sigma_i^2(\Zm)) \\
&+ 2 \sigma_{\Zm}^2 \textnormal{div}(D_{\tau}(\Zm)),\\
\end{split}
\end{equation}
where $\sigma_i(\Zm)$ is the $i$-th singular value of $\Zm$ for $i \in \{1,2,\ldots, N\}$. 
A closed-form expression for the divergence of this estimator is obtained in \cite{CST13}. For the case in which $\Zm\in\mathbb{R}^{2M\times N}$, the divergence is given by
\begin{equation} \label{div}
\ssqeq
\begin{split}
\textnormal{div}(D_{\tau}(\Zm)) = & \sum_{i=1}^{N} \bigg[\mathds{1}(\sigma_i(\Zm)>\tau)+ (2M-N)\frac{(\sigma_i(\Zm)-\tau)_+}{\sigma_i(\Zm)} \bigg]\\ 
& +2 \sum_{ i \neq j,i,j=1}^{N} \frac{\sigma_i(\Zm) (\sigma_i(\Zm)-\tau)_+}{\sigma_i^2(\Zm)-\sigma_j^2(\Zm)},
\end{split}
\end{equation}
where $\mathds{1}(\cdot) $ denotes the indicator function. In addition, when $\Zm$ has repeated singular values, the divergence in (\ref{div}) is set to be zero.
The proposed algorithm approximates $\sigma_{\Zm}^2$ with the weighted sum of the noise in $\Omega$ and in $\Omega^c$.
A detailed description of the BSVT algorithm is presented in \cite{GEPOC17} but we reproduce the algorithm below to aid with the presentation. Note that ${D}_{\textnormal{LMMSE}}$ represents the average noise per entry in $\Omega^c$.
The main advantage of the BSVT algorithm is that the threshold is optimized at each iteration. This is achieved by incorporating the prior knowledge about the matrix in the form of second order statistics via the introduction of the SURE and LMMSE steps. Admittedly, this approach requires additional knowledge that is not necessary when using the SVT algorithm. However, it is shown in \cite{GEPOC17} that the introduction of the prior knowledge enables a robust recovery of the missing entries.
\begin{algorithm}
\caption{Bayesian Singular Value Thresholding}\label{alg:bsvt}
\begin{algorithmic}[1]
\Require set of observations $\Omega$, observed entries $P_{\Omega}(\Rm)$, mean $\mathbf{0}$, covariance matrix $\Sigmam$, step size $\delta_b$, tolerance $\epsilon$, and maximum iteration count $k_{\textnormal{max}}$
\Ensure $\widehat{\Mm}_{\textnormal{BSVT}}$
\State Set $\Ym^0=\mathbf{0}$
\State Set $\Zm^0=\mathbf{0}$
\State Set $\tau=0$
\State Set $\Omega^c=\{1,2,...,2M\} \times \{1,2,...,N\}\setminus \Omega$
\For {$k=1$ to $k_{\textnormal{max}}$}
\State Compute $[\Um, \Sm, \Vm]=\textnormal{svd}(\Zm^{(k-1)})$
\State Set $\Xm^{(k)}= \sum_{j=1}^{N} \textnormal{max}(0,\sigma_j(\Zm^{(k-1)})-\tau^{(k-1)})\uv_j \vv_j$
\If  {$\|P_{\Omega}(\Xm^{(k)}-\Rm)\|_F / \|P_{\Omega}(\Rm)\|_F \leq \epsilon$} {\bf break} \EndIf
\State Set $\Ym^{(k)}=\Ym^{(k-1)}+ \delta_b \big( P_{\Omega}(\mathbf{R})-P_{\Omega}(\mathbf{X}^{(k)}) \big)$
\State Set $\Lm^{(k)}=\Sigmam_{\Omega^c \Omega} \Sigmam_{\Omega \Omega}^{-1} \Ym^{(k)}$
\State Set $\Zm^{(k)}=\Ym^{(k)}+\Lm^{(k)}$
\State Set \begin{small}$\sigma_{\Zm^{(k)}}^2 = (\|\Ym^{(k)}-P_{\Omega}(\Rm)\|_F^2 + |\Omega^c| {D}_{\textnormal{LMMSE}})/2MN$  \end{small}
\State Set $\tau^{(k)}=\underset{\tau}{\textnormal{arg\,min}} \, \textnormal{SURE}(D_{\tau})(\Zm^{(k)})$ 
\EndFor
\State Set $\widehat{\Mm}_{\textnormal{BSVT}}= \Xm^{(k)}$
\end{algorithmic}
\end{algorithm}  \vspace{-1mm}

\section{Joint recovery of missing data in two datasets}

In this section, an estimation framework for recovering missing data from different datasets is proposed. The estimation framework facilitates considering the correlation between datasets for a wide range of correlation structures. We begin by noting that in the joint recovery case, there are two types of correlation between the entries of the combined matrix. First, the intra-correlation that refers to the correlation between the entries within each dataset. This is the type of correlation that is exploited in the independent recovery scenario, i.e., when the missing entries from each dataset are recovered using only available observations from that dataset. Second, the cross-correlation defined as the correlation between the data points from the two different datasets. In contrast to the independent recovery case, a joint recovery technique needs to account for both types of correlation. By considering the cross-correlation, the recovery process leverages on other types of data in order to recover the datasets with limited available observations.

In an MC setting, the minimum number of observations required depends on the size and the rank of the matrix \cite{CR09}. The combination of the datasets into a single matrix increases the size of the matrix, and therefore, the fundamental limit for the joint recovery case depends on the tradeoff between the size and the rank of the combined matrix, and the number of observations available for each dataset.
Note that the rank of the combined matrix depends on both the intra and the cross-correlation.
The following lemma provides lower and upper bounds for the rank of the combined matrix based on the individual rank of the matrices. To that end, let us denote rank of the matrices by $\textnormal{rank}(\Mm_1)=r_1$, $\textnormal{rank}(\Mm_2)=r_2$, and $\textnormal{rank}(\Mm)=r$.

\begin{lemma} \label{bounds_on_rank_M}
Let  $\Mm_1 \in \mathbb{R}^{M \times N}$ and $\Mm_2 \in \mathbb{R}^{M \times N}$. Define the combined matrix $\Mm=\left[
\begin{array}{c}
\Mm_1\\
\Mm_2
\end{array}
\right] \in \mathbb{R}^{2M \times N}$. Then, the following holds:
\begin{equation} \label{eq:lemma1}
\textnormal{max}(r_1,r_2) \leq r \leq r_1+r_2.
\end{equation}
\end{lemma}

\begin{proof}
The rank of the matrix $\Mm$ is defined as (3.23(a) in \cite{seber})
\begin{equation} \label{rank_M_seber}
r=r_1+r_2-d,
\end{equation}
where $d$ is the intersection of the row subspaces of the matrices $\Mm_1$ and $\Mm_2$ given by
\begin{equation}
d=\textnormal{dim}[{\Cc}(\Mm_1^{\sf T}) \cap  {\Cc}(\Mm_2^{\sf T})],
\end{equation}
where $\Cc(\Mm)$ denotes the column subspace of the matrix $\Mm$.
Since the intersection of the row subspaces is bounded by
\begin{equation}
0 \leq d \leq \textnormal{min}(r_1,r_2),
\end{equation}
the rank of the matrix $\Mm$ satisfies (\ref{eq:lemma1}).
\end{proof}

Based on the insight provided by Lemma \ref{bounds_on_rank_M}, the intra-correlation determines the rank of the matrices $\Mm_1$ and $\Mm_2$ which define the lower and upper bounds on $r$. A smaller value of intra-correlation in one of the datasets results in a larger lower bound for $r$. On the other hand, the cross-correlation governs the value of $r$ within the limits defined by Lemma \ref{bounds_on_rank_M}. Indeed, a larger value of cross-correlation results in a value of $r$ that is closer to the lower bound while a smaller value of cross-correlation generates a combined matrix with a rank that is closer to the upper bound.
In other words, the intra-correlation defines the limit values of $r$ for which recovery is feasible and the cross-correlation governs the value of $r$ within the limit.

In \cite{RSB15} it is shown that the low rank matrices $\Mm_1$ and $\Mm_2$ can be successfully recovered independently when the number of observations for the first matrix, denoted by $k_1$, satisfies
\begin{equation} \label{eq:bound_k1}
k_1 > (M+N-r_1)r_1,
\end{equation}
and the number of available observations for the second matrix, denoted by $k_2$, obeys
\begin{equation} \label{eq:bound_k2}
k_2 > (M+N-r_2)r_2.
\end{equation}
This result is based on the assumption that for the random matrices $\Mm_1$ and $\Mm_2$ there exist the $\sigma$-measures $\mu_1$ and $\mu_2$, respectively, and that both measures admit a Lebesgue decomposition.
For the combined matrix $\Mm$, the $\sigma$-measure is obtained as the product of the measures of $\Mm_1$ and $\Mm_2$ \cite{halmos13}
\begin{equation} 
\mu=\mu_1 \times \mu_2.
\end{equation}
Moreover, since $\mu$ is a $\sigma$-measure it also admits a Lebesgue decomposition \cite{HS13} and \cite{royden88}. Hence, the result in \cite{RSB15} applies for the combined matrix $\Mm$ without any additional assumptions, i.e.,
\begin{equation} \label{eq:bound_k}
k_1+k_2 > (2M+N-r)r.
\end{equation}
\begin{figure}[t!]
\centering
\includegraphics[width=0.45\textwidth]{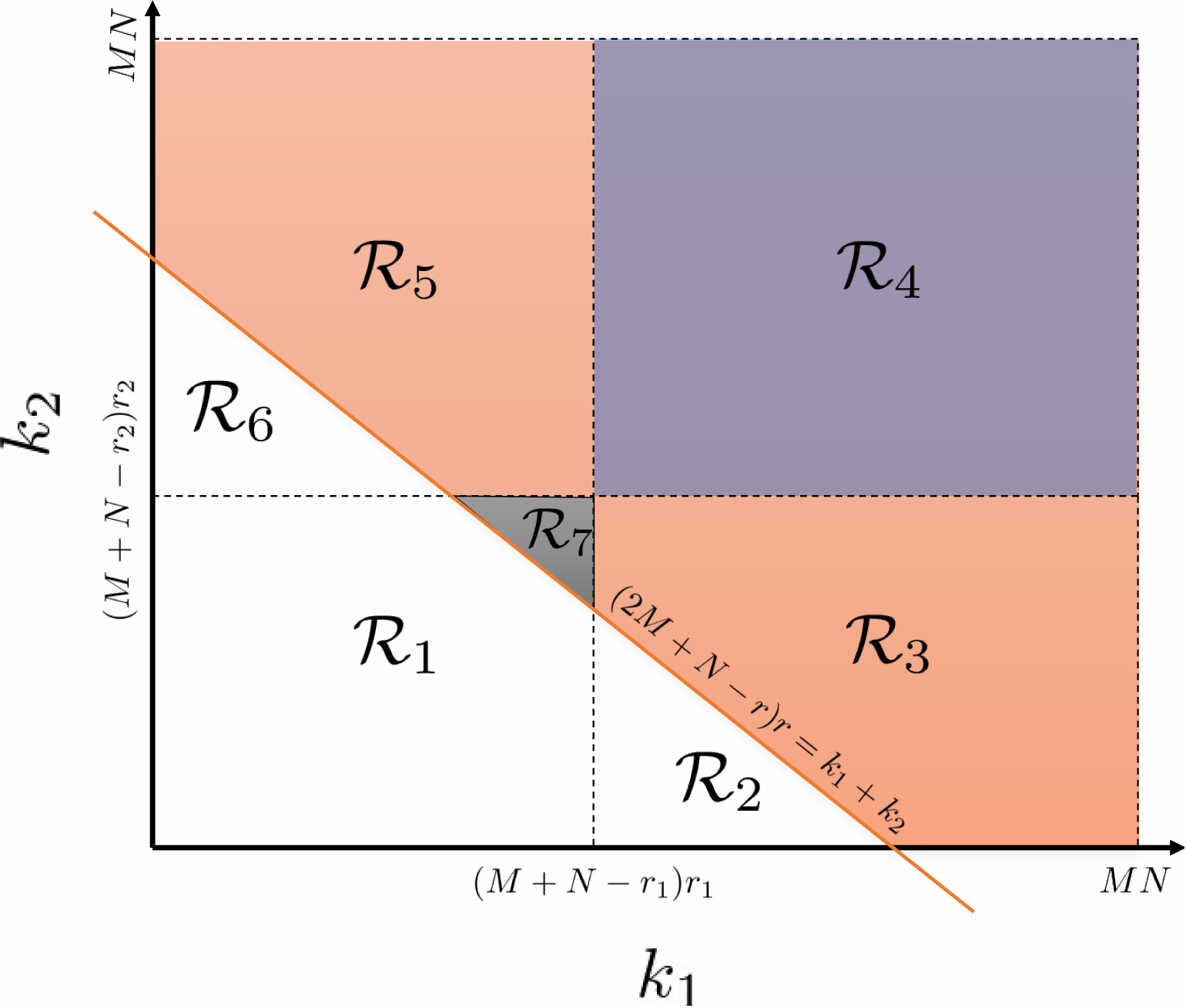}
\caption{Example of recovery regions imposed by the fundamental limit for the recovery of the matrices $\Mm_1$, $\Mm_2$ and $\Mm$.}
\label{fig:polytope_eq_w_gain}
\vspace{-5mm}
\end{figure}
Fig. \ref{fig:polytope_eq_w_gain} depicts the inequalities in (\ref{eq:bound_k1}), (\ref{eq:bound_k2}) and (\ref{eq:bound_k}) that describe the lower bound on the number of observations required to recover the matrices $\Mm_1$, $\Mm_2$ and $\Mm$, respectively. 
The bounds divide the $(k_1,k_2)$ plane into seven regions that correspond to different recovery scenarios for the independent and joint estimation settings. The seventh region, i.e., $\Rc_7$, corresponds to the case in which the independent recovery of each dataset is not possible but the joint recovery is feasible. In other words, the existence of region $\Rc_7$ is equivalent to the case in which it is beneficial to jointly recover the two datasets. This scenario is captured by the following definition.
\begin{definition} \label{def:bjr_using_r7}
The joint recovery of two matrices, $\Mm_1$, $\Mm_2 \in \mathbb{R}^{M \times N}$ of rank $r_1$ and $r_2$, respectively, is beneficial in the region given by
\begin{equation}
\sqeq
\begin{split}
\Rc_7=\big\{&(k_1,k_2) \in \mathbb{N}^2 : k_1 \leq (M+N-r_1)r_1, \\
 &k_2 \leq (M+N-r_2)r_2, \,  k_1+k_2 > (2M+N-r)r\big\}.
\end{split}
\end{equation}
\end{definition}
Note that when the rank of the combined matrix increases, the line described by $(2M+N-r)r=k_1+k_2$ is shifted towards larger values which for values of $r$ larger than a given threshold induces an empty region $\Rc_7$. The value of the threshold is given by the sufficient condition in Theorem \ref{th:NS_for_bjr}.
It is also worth noting that based on the regions depicted in Fig. \ref{fig:polytope_eq_w_gain}, the number of cases in which the joint recovery is feasible is larger than the number of cases in which the independent recovery of the datasets is possible.
The following theorem provides the necessary and sufficient conditions to guarantee that the joint recovery is beneficial.
\begin{theorem} \label{th:NS_for_bjr}
Let $\Mm_1, \Mm_2 \in \mathbb{R}^{M \times N}$, with rank $r_1$ and $r_2$. Then, the joint recovery of the two matrices requires fewer observations than the independent recovery if
\begin{equation} \label{lemma1_cond}
\begin{split}
1- \frac{\textnormal{max}(r_1,r_2)}{\textnormal{min}(r_1,r_2)} > 
\frac{\textnormal{min}(r_1,r_2)-N}{M},
\end{split}
\end{equation} 
and the rank of the combined matrix satisfies
\begin{equation} \label{eq_th_bjr}
\begin{split}
r <& M+\frac{1}{2}N - \frac{1}{2}(M+N-2r_1-2r_2) \\
&\bigg(1+ \frac{3M^2+2MN-8r_1r_2}{(M+N-2r_1-2r_2)^2}\bigg)^{1/2}.
\end{split}
\end{equation} 
\end{theorem}
\begin{proof}
The proof for the necessary condition in (\ref{lemma1_cond}) hinges on the fact that the quadratic inequality on $r$ that describes the cases in which the joint recovery is beneficial, given by
\begin{equation} \label{eq:quadratic_th}
\sqeq
(2M+N-r)r < (M+N-r_1)r_1 + (M+N-r_2)r_2,
\end{equation}
is satisfied for a value bounded by Lemma \ref{bounds_on_rank_M}. This is achieved by showing that the smaller root of the quadratic is contained in the open interval $\big(\textnormal{max}(r_1,r_2),r_1+r_2\big)$.

The sufficient condition in (\ref{eq_th_bjr}) is equivalent to showing that the smaller root of the quadratic in (\ref{eq:quadratic_th}) is a strict upper bound for $r$.
\end{proof}
Note that the necessary condition from Theorem \ref{th:NS_for_bjr}, that is the inequality in (\ref{lemma1_cond}) depends only on the matrices $\Mm_1$ and $\Mm_2$ and does not depend on the combined matrix $\Mm$. In contrast, the sufficient condition in (\ref{eq_th_bjr}) provides an upper bound for the rank of the combined matrix such that the total number of observations required for the joint recovery is fewer when compared to the independent recovery case. Theorem \ref{th:NS_for_bjr} provides a necessary and sufficient condition for the joint recovery of two data matrices to be beneficial.
\section{Numerical results} \label{sec:numres}

This section presents a numerical evaluation of the joint recovery performance for two datasets. The matrices used for the simulations are generated using the model described in Section \ref{sec:syn_data_model} and the size of the matrices $\Mm_1$ and $\Mm_2$ is fixed such that $M=50$ and $N=100$, respectively. Hence, the joint matrix $\Mm$ is a square matrix of size $100$.
The range of rank values selected aims to characterize the joint recovery in two scenarios: when the two combined matrices have the same rank, i.e., $r_1=6$ and $r_2=6$, and when the ratio between the two rank values is small, i.e., $r_1=6$ and $r_2=9$. 
To facilitate this, the synthetic data model presented in Section \ref{sec:syn_data_model} is used to generate correlated data matrices with the rank values of interest.
\subsection{Simulation framework} 

\begin{figure}[t!]
\centering
\includegraphics[width=0.5\textwidth]{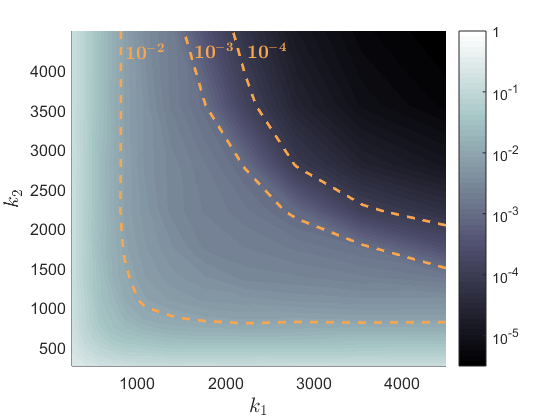}
\caption{Joint recovery error using BSVT, measured by NMSE, when $r_1=6$, $r_2=6$, $r=9$ and SNR=$50$ dB.}
\label{fig:bsvt_snr50_6_6_9}
\end{figure}

\begin{figure}[t!]
\centering
\includegraphics[width=0.5\textwidth]{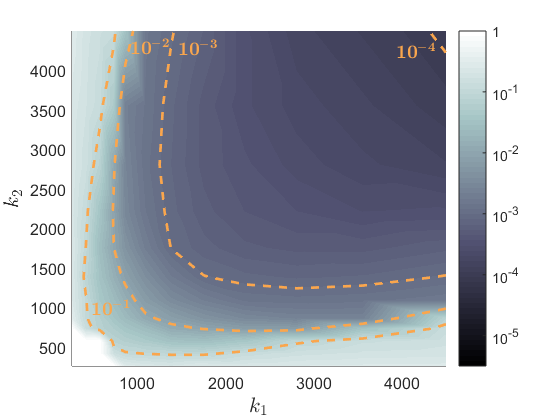}
\caption{Joint recovery error using SVT, measured by NMSE, when $r_1=6$, $r_2=6$, $r=9$ and SNR=$50$ dB.}
\label{fig:svt_snr50_6_6_9}
\end{figure}

\begin{figure}[t!]
\centering
\includegraphics[width=0.5\textwidth]{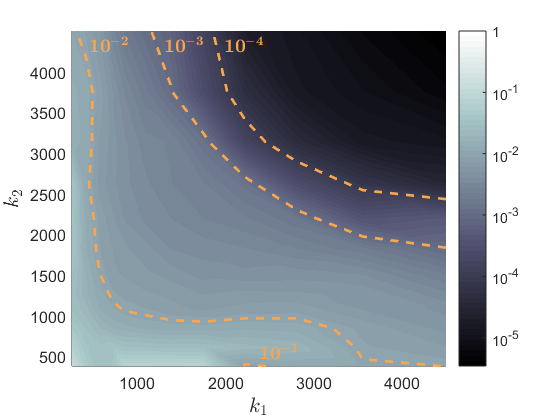}
\caption{Joint recovery error using BSVT, measured by NMSE, when $r_1=6$, $r_2=9$, $r=10$ and SNR=$50$ dB.}
\label{fig:bsvt_snr50_6_9_10}
\end{figure}

\begin{figure}[t!]
\centering
\includegraphics[width=0.5\textwidth]{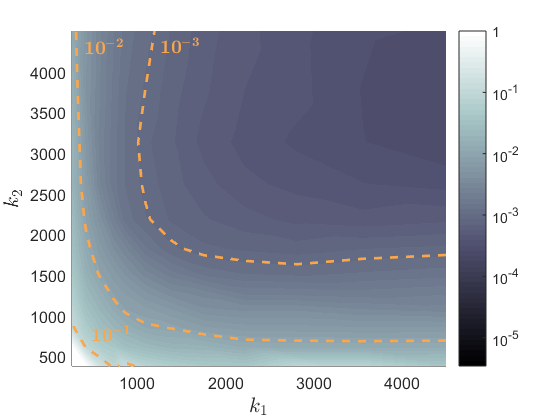}
\caption{Joint recovery error using SVT, measured by NMSE, when $r_1=6$, $r_2=9$, $r=10$ and SNR=$50$ dB.}
\label{fig:svt_snr50_6_9_10}
\end{figure}

Using the mathematical model defined in Section \ref{sec:syn_data_model}, the covariance matrix for the combined matrix $\Mm$ is given by (\ref{eq:synth_data_model}) where the intra-correlation between the state variables in $\Mm_1$ is modelled by the $\upsilon_{11}$ parameter in (\ref{eq:end_corr}), the intra-correlation between the state variables in $\Mm_2$ is modelled by the $\upsilon_{22}$ parameter in (\ref{eq:end_corr}) and the cross-correlation between $\Mm_1$ and $\Mm_2$ is modelled by $\psi$ in (\ref{eq:synth_data_model}). The numerical analysis shows that a larger value of $\upsilon_{ll}$ results in a more correlated matrix $\Mm_l$ and consequently a smaller value for $r_l$, where $l \in \{1,2\}$. Moreover, the cross-correlation between $\Mm_1$ and $\Mm_2$ increases with the value of $\psi$ which leads to a smaller value for $r$ within the bounds defined by Lemma \ref{bounds_on_rank_M}.

The matrix $\Mm$ generated using the model in (\ref{eq:synth_data_model}) is not exactly low rank. Instead, it can be well approximated by a low rank matrix. 
Let us denote by $\widetilde{\Mm}(r)$ the low rank approximation of rank $r$ obtained by vanishing the smallest $N-r$ singular values of the matrix $\Mm$. 
In the following, $r$ is defined as the minimum value for which the NMSE between the matrix $\Mm$ and the low rank approximation of rank $r$, i.e., $\widetilde{\Mm}(r)$, is below $10^{-3}$.
Consequently, the model in (\ref{eq:synth_data_model}) is used to generate data matrices $\Mm$ such that the low rank approximations $\widetilde{\Mm}_1(r_1)$, $\widetilde{\Mm}_2(r_2)$ and $\widetilde{\Mm}(r)$ have the intended ranks. Moreover, the low rank approximation of the combined matrix, i.e., $\widetilde{\Mm}(r)$, is used to evaluate the numerical performance of both BSVT and SVT in exploiting the correlation between the two datasets.
The recovery performance of both algorithms is averaged over ten realizations of $\Omega$, where the locations of the available entries are sampled uniformly at random in each dataset.

In the following, a numerical analysis for the joint recovery performance of SVT and BSVT is presented for the cases in which the rank values for the combined matrices are: $r_1=6$, $r_2=6$, $r=9$ and $r_1=6$, $r_2=9$, $r=10$. 
The choice of rank for the combined matrix resembles a high cross-correlation case in which the value of $r$ satisfies the condition imposed by Theorem \ref{th:NS_for_bjr}. We consider a low noise regime for which SNR=$50$ dB in both datasets to emphasize the impact of the intra and cross-correlation in the recovery process, where the SNR in dataset $l \in \{1, 2\}$ is defined as
\begin{equation} 
\textnormal{SNR}_l \eqdef 10 \textnormal{log}_{10} {\frac{\frac{1}{M}\textnormal{Tr}(\Sigmam_{ll})}{\sigma_{\Nm_l}^2}},
\end{equation}
where $\Sigmam_{ll}$ is defined in (\ref{eq:sigma_ll}) and $\sigma_{\Nm_l}^2$ is described in (\ref{eq:noise}). 
A wider range of rank values and noise regimes is presented in \cite{genes_phd}.
The efficiency in exploiting the cross-correlation between the combined datasets is evaluated by comparing the recovery performance across different sampling regimes in which the number of observations in $\Omega$ is constant but the ratio between the number of available entries in each dataset varies. The cross-correlation is successfully exploited when the recovery error is similar across different sampling regimes.


Fig. \ref{fig:bsvt_snr50_6_6_9} depicts the performance of the BSVT algorithm when $r_1=6$, $r_2=6$, $r=9$ and SNR=$50$ dB. 
Interestingly, the contour lines for the $10^{-4}$ and $10^{-3}$ recovery error exhibit a similar shape to the line depicted by $(2M+N-r)r=k_1+k_2$ in Fig. \ref{fig:polytope_eq_w_gain}. This suggests that the BSVT algorithm successfully exploits the cross-correlation in that region and obtains a similar recovery performance tradeoff when the ratio between $k_1$ and $k_2$ varies for a fixed value of $k_1+k_2$.
In contrast, the contour lines for $10^{-3}$, $10^{-2}$ and $10^{-1}$ SVT recovery error depicted in Fig. \ref{fig:svt_snr50_6_6_9} exhibit a similar shape to the region $\Rc_4$ in Fig. \ref{fig:polytope_eq_w_gain} which corresponds to the independent recovery area in which the cross-correlation is not exploited. Based on this observation, it is reasonable to assume that SVT is not effective in exploiting the cross-correlation as the recovery error changes with the ratio between $k_1$ and $k_2$ for a fixed total number of observations.


Fig. \ref{fig:bsvt_snr50_6_9_10} depicts the performance of the BSVT algorithm when $r_1=6$, $r_2=9$, $r=10$ and SNR=$50$ dB. In line with the case discussed in Fig. \ref{fig:bsvt_snr50_6_6_9}, the contour lines for the $10^{-4}$ and $10^{-3}$ recovery error exhibit a similar shape to the fundamental limit in Fig. \ref{fig:polytope_eq_w_gain}. This suggests that the BSVT approach is able to exploit the cross-correlation between the combined datasets in the almost noiseless regime for both rank cases considered.
In Fig. \ref{fig:svt_snr50_6_9_10} the performance of the SVT algorithm is depicted for the case in which $r_1=6$, $r_2=9$, $r=10$ and SNR=$50$ dB.
In this case, the shape of the contour lines for $10^{-3}$ and $10^{-2}$ recovery error is similar to the shape of the region $\Rc_4$ in Fig. \ref{fig:polytope_eq_w_gain} which suggests that the SVT algorithm is not efficient in exploiting cross-correlation. Consequently, the BSVT algorithm is able to exploit the cross-correlation between the combined datasets more effectively when compared to the SVT approach in the almost noiseless regime. 
The gain in recovery performance is facilitated by the prior knowledge incorporated in the structure of the BSVT algorithm.

\section{Conclusion}
The fundamental limits for the joint recovery of two datasets have been characterized in terms of the rank of the single and combined data matrices. Theoretical conditions are derived for the case in which the joint recovery of two datasets requires less observations compared to the independent recovery case. Based on the insight provided by the fundamental limit, the number of cases in which the joint recovery is feasible is significantly larger when compared to the independent recovery setting.

A model for correlated datasets is proposed. Numerical results show that the correlation between different types of data is exploited by leveraging the information provided by the dataset with fewer missing entries to enable the recovery of the other dataset.
Moreover, in contrast to the SVT algorithm, the performance of the BSVT approach matches the geometry imposed by the fundamental limit which suggests that BSVT is indeed better suited to exploit the correlation between datasets.




\bibliographystyle{IEEE}
\balance
\bibliography{references}

%
%
%
\end{document}